\documentclass[11pt]{article}
\usepackage[bookmarks]{hyperref}
\usepackage{
	amsfonts,
	amsmath,
	amssymb,
	amsthm,
	array,%Kayleigh
	caption,
	chngpage,%Kayleigh
	color,
	complexity,%\NP and the like
	float,%Kayleigh
	graphicx,
	latexsym,
	mathrsfs,
	multicol,%Kayleigh
	multirow,%Kayleigh
	sectsty,
	subcaption,
	tikz,
	tikz-cd,
	verbatim,
	wrapfig,
}

\usetikzlibrary{arrows,automata, positioning}
\usetikzlibrary{decorations.pathmorphing}
\tikzset{snake it/.style={decorate, decoration=snake}}
\allsectionsfont{\centering}

\newcommand{\eps}{\varepsilon}
\newcommand{\abs}[1]{\lvert#1\rvert}

\newtheorem{thm}{Theorem}

\newtheorem{lem}[thm]{Lemma}

\newtheorem{con}[thm]{Conjecture}
\newtheorem{df}[thm]{Definition}
\newtheorem{rem}[thm]{Remark}

\newcommand{\myUrl}[1]{
	\begin{center}
		{\small\url{#1}} 
	\end{center}
}

\title{
	On the complexity of automatic complexity
}
\author{
	Bj{\o}rn Kjos-Hanssen\footnote{
		This work was partially supported by
		a grant from the Simons Foundation (\#315188 to Bj\o rn Kjos-Hanssen).
		This material is based upon work supported by the National Science Foundation under Grant No.\ 1545707.
	}
}
\begin{document}
	\maketitle{}
	\begin{abstract}
		Generalizing the notion of
		automatic complexity of individual words due to Shallit and Wang,
		we define the automatic complexity $A(E)$ of
		an equivalence relation $E$ on a finite set $S$ of words.

		We prove that
		the problem of determining whether $A(E)$ equals the number $\abs{E}$ of equivalence classes of $E$ is \NP-complete.
		The problem of determining whether $A(E) = \abs{E} + k$ for a fixed $k\ge 1$ is
		complete for the second level of the Boolean hierarchy for \NP, i.e., $\BH_2$-complete.
		
		Let $L$ be the language consisting of all words of
		maximal nondeterministic automatic complexity.
		We characterize the complexity of infinite subsets of $L$ by showing that they can be co-context-free
		but not context-free, i.e., $L$ is \CFL-immune, but not \co\CFL-immune.

		We show that for each $\eps>0$, $L_\eps\not\in\co\CFL$, where
		$L_\eps$ is the set of all words whose deterministic automatic complexity $A(x)$ satisfies
		$A(x)\ge \abs{x}^{1/2-\eps}$.
	\end{abstract}
	%\tableofcontents
	\section{Introduction}
		Automatic complexity was introduced by Shallit and Wang \cite{MR1897300}
		as a way to retain some of the power of Kolmogorov complexity while
		obtaining a computable notion.
		They raised the question whether the automatic complexity of a string (which we shall call a \emph{word}) $x$ is in fact
		polynomial-time computable as a function of $\abs{x}$, the length of $x$.
		We give partial negative results for that question in two ways.
		(Our results also partially address Allender's question \cite[Open Question 3.8]{allender} whether there is evidence that automatic complexity is computationally intractable.)
		\begin{itemize}
			\item For nondeterministic automatic complexity,
				introduced by Hyde and Kjos-Hanssen \cite{Kjos-EJC},
				there is a natural notion of maximally complex words.
				We show that the language $L$ consisting of all such words is not context-free,
				by virtue of being \CFL-immune.
				This result appears to be at the right level of the complexity hierarchy, insofar as
				we also show that $L$ is not \co\CFL-immune.
				While we do not know whether $L\in\co\CFL$,
				a related language consisting of ``somewhat complex'' words is shown to be non-\co\CFL.
			\item We generalize automatic complexity to a more general notion of automatic complexity
				of equivalence relations on words, and show that is not polynomial-time computable.
				In particular, we show that
				the set of minimally complex equivalence relations is \NP-complete and
				the set of equivalence relations whose complexity is exactly a constant $k$ above
				the minimum is $\BH_2$-complete.
		\end{itemize}

		In the past,
		Gold \cite{MR0495194} and Angluin \cite{MR523447}
		established \NP-completeness for related problems.
		Heggernes et al.~\cite{MR3305960} considered parametrized complexity variations, such as
		fixing the number of states at two ($\abs{Q}=2$) and increasing the alphabet size.

		As an illustration of the power and computability of automatic complexity, we have created the following web service.
		To find the complexity of, say, the word $01011010$, and an illustration of any automaton used in the associated proof,
		go to
		\myUrl{http://math.hawaii.edu/wordpress/bjoern/complexity-of-01011010/}
		Alternatively, play the \textsc{Complexity Guessing Game} at:
		\myUrl{http://math.hawaii.edu/wordpress/bjoern/complexity-guessing-game/}

	\section{The set of maximally complex words is \CFL-immune but not \co\CFL-immune}
		\begin{df}[{\cite{Kjos-EJC, MR1897300}}]
			Let $x$ be a finite word. The nondeterministic automatic complexity $A_{\mathrm N}(x)$ of $x$ is the minimum number of states of a nondeterministic
			finite automaton that accepts $x$, and does not accept any other word of length $\abs{x}$,
			and accepts $x$ via only one computation path.

			The (deterministic) automatic complexity $A(x)$ of $x$ is the minimum number of states of a deterministic
			finite automaton that accepts $x$, and does not accept any other word of length $\abs{x}$.
		\end{df}
		\begin{thm}[Hyde \cite{Kjos-EJC}]\label{again}
			For a word $x$ of length $n$,
			\[
				A_{\mathrm N}(x) \le \left\lfloor \frac{n}2\right\rfloor + 1.
			\]
		\end{thm}
		An idea of the proof of Theorem \ref{again} is given in Figure \ref{fig1}.
		\begin{figure}%[h]
			\centering
			\begin{tikzpicture}[shorten >=1pt,node distance=1.5cm,on grid,auto]
				\node[state,initial, accepting] (q_1)   {$q_1$}; 
				\node[state] (q_2)     [right=of q_1   ] {$q_2$}; 
				\node[state] (q_3)     [right=of q_2   ] {$q_3$}; 
				\node[state] (q_4)     [right=of q_3   ] {$q_4$};
				\node        (q_dots)  [right=of q_4   ] {$\ldots$};
				\node[state] (q_m)     [right=of q_dots] {$q_m$};
				\node[state] (q_{m+1}) [right=of q_m   ] {$q_{m+1}$}; 
				\path[->] 
					(q_1)     edge [bend left]  node           {$x_1$}     (q_2)
					(q_2)     edge [bend left]  node           {$x_2$}     (q_3)
					(q_3)     edge [bend left]  node           {$x_3$}     (q_4)
					(q_4)     edge [bend left]  node [pos=.45] {$x_4$}     (q_dots)
					(q_dots)  edge [bend left]  node [pos=.6]  {$x_{m-1}$} (q_m)
					(q_m)     edge [bend left]  node [pos=.56] {$x_m$}     (q_{m+1})
					(q_{m+1}) edge [loop above] node           {$x_{m+1}$} ()
					(q_{m+1}) edge [bend left]  node [pos=.45] {$x_{m+2}$} (q_m)
					(q_m)     edge [bend left]  node [pos=.4]  {$x_{m+3}$} (q_dots)
					(q_dots)  edge [bend left]  node [pos=.6]  {$x_{n-3}$} (q_4)
					(q_4)     edge [bend left]  node           {$x_{n-2}$} (q_3)
					(q_3)     edge [bend left]  node           {$x_{n-1}$} (q_2)
					(q_2)     edge [bend left]  node           {$x_n$}     (q_1);
			\end{tikzpicture}
			\caption{
				A nondeterministic finite automaton that only accepts one word
				$x= x_1 x_2 x_3 x_4 \cdots x_n$ of length $n = 2m + 1$.
			}
			\label{fig1}
		\end{figure}
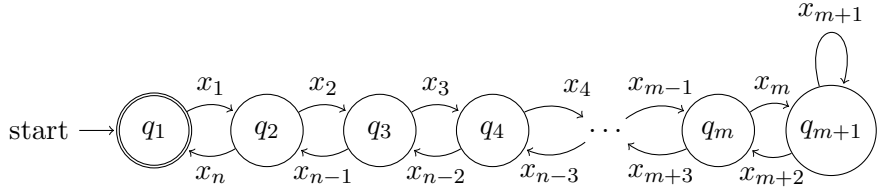
		\begin{df}
			Let $b(n) = \lfloor \frac{n}2\rfloor + 1$ be the canonical upper bound for $A_{\mathrm N}$ from Theorem \ref{again}.
			Let $L_k = \{x \in \{0, 1, \dots, k - 1\}^*: A_{\mathrm N}(x) = b(n)\}$.
			Any $x\in L_k$ is called a \emph{maximally complex} word.
		\end{df}
		\begin{rem}
			$L_3$ is known to be infinite (see Theorem \ref{berstel}) but we do not know whether $L_2$ is infinite.
		\end{rem}
		\begin{lem}\label{bahnhof}
			Let $x_0$, $y_0$, $a$, $b$ be positive integers with $a$ and $b$ relatively prime, $x_0<b$, and $y_0<a$.
			Then the equation
			\begin{equation}
				a x + b y = a x_0 + b y_0\label{bonnie}
			\end{equation}
			has a unique solution $(x,y)$ in nonnegative integers.
		\end{lem}
		\begin{proof}
			Equation (\ref{bonnie}) implies
			\[
				a(x - x_0)\equiv 0\pmod{b}.
			\]
			Since $a$ and $b$ are relatively prime it follows that $x-x_0\equiv 0\pmod{b}$.
			Thus $x=x_0 + n b$ for some $n\in\mathbb Z$.
			If $n<0$ then $x\le x_0 - b < 0$, which contradicts the requirement that $x\ge 0$.
			If $n>0$ then using $y\ge 0$,
			\[
			 	ax + by \ge a(x_0 + b) + b(0) > ax_0 + by_0
			\]
			contradicting (\ref{bonnie}). Thus $n=0$ and the only solution is $x=x_0$.
		\end{proof}
		\begin{df}%[\cite{MR2883055}]
			For any collection of languages $\mathsf{M}$, a language $L$ is $\mathsf{M}$-immune if
			it is infinite and contains no infinite subset in $\mathsf{M}$.
			We say that $L\in\co\mathsf{M}$ if the complement of $L$ belongs to $\mathsf{M}$.
			Let \CFL\ be the class of all context-free languages.
		\end{df}
		\begin{thm}[Pumping lemma for \CFL\ {\cite{MR0151376}}]\label{pump}
			If a language $L$ is context-free, then there exists some integer $p\ge 1$ (a \emph{pumping length}) such that
			every word $s$ in $L$ with $\abs{s}\ge p$ can be written as $s = uvwxy$
			where
			\begin{enumerate}
				\item $\abs{vwx}\le p$,
				\item $\abs{vx}\ge 1$, and
				\item $uv^Nwx^Ny$ is in $L$ for all $N\ge 0$.
			\end{enumerate}
		\end{thm}
		\begin{df}\label{morphismDef}
			Let $\Sigma$ be a finite alphabet.
			A function $\pi:\Sigma^*\to\Sigma^*$ is a \emph{homomorphism}
			if it respects concatenation: for all $x$, $y$,
			\[
				\pi(xy)=\pi(x)\pi(y).
			\]
		\end{df}
		\begin{thm}\label{berstel}
			$L_3$ is not \co\CFL-immune.
		\end{thm}
		\begin{proof}
			Let $\mathbf t$ be an
			infinite square-free word over $\{0,1,2\}$ generated by, and a fixed point of, a homomorphism.
			Such a $\mathbf t$ was constructed by
			Thue \cite{ThueTwo}.
			Let
			\[
				\mathrm{Pref}(\mathbf t) = \{x : x \text{ is a prefix of }\mathbf t\}.
			\]
			By Berstel \cite[Theorem on page 7]{MR825639}, $\mathrm{Pref}(\mathbf t)\in \co\CFL$.
			Since by \cite[Theorem 18]{Kjos-EJC} every square-free word over $\{0,1,2\}$ belongs to $L_3$,
			we also have $\mathrm{Pref}(\mathbf t)\subseteq L_3$.
		\end{proof}
		\begin{thm}\label{context}
			$L_3$ is \CFL-immune.
		\end{thm}
		\begin{proof}
			Since $L_3$ is not \co\CFL-immune (Theorem \ref{berstel}), in particular $L_3$ is infinite.
			Suppose $L_3$ has an infinite subset $K\in\CFL$.
			By the pumping lemma (Theorem \ref{pump})
			there is a ``pumping length'' $p$ such that
			any word $X\in K$ of length at least $p$
			can be written as $X = uvwxy$, where $\abs{vx}\ge 1$ and
			\[
				X_N := uv^Nwx^Ny\in K\subseteq L_3\qquad\text{for all }N \ge 0.
			\]
			We denote the length of $X_N$ by $n_N$.
			Since $L_3$ is infinite, there exists at least one such word $X$.

			\emph{Case 1:} $\abs{v} \ne \abs{x}$.
				Let us first assume $\abs{v}>\abs{x}$.
				In particular, $\eps := \frac{\abs{v}}{\abs{v}+\abs{x}}-\frac12>0$.
				Consider an automaton which loops at each occurrence of $v$ and otherwise proceeds to the right (Figure \ref{contextFig1}).
				Let $N$ be so large that $\abs{v^N}\ge\left(\frac12+\frac{\eps}2\right)n_N$
				and $\frac{\abs{v}}{n_N}\le\frac{\eps}2$.
				Then
				\[
					A_{\mathrm N}(X_N)\le \abs{uwx^Ny}+\abs{v} = n_N-\abs{v^N} + \abs{v}\le n_N - \left(\frac12+\frac{\eps}2\right)n_N + \abs{v}
				\]
				\[
					= \left(\frac12-\frac{\eps}2\right)n_N + \abs{v}  \le \frac{n_N}2.
				\]
				and so $X_N\not\in L_3$.

		\begin{figure}
			\centering
			\begin{tikzpicture}[shorten >=1pt,node distance=2.5cm,on grid,auto]
				\node[state,initial]   (q_0)                {$q_0$};
				\node[state]           (q_1) [right of=q_0] {$q_1$};
				\node[state]           (q_2) [right of=q_1] {$q_2$};
				\path[->]
					(q_0) edge                 node {$u$}   (q_1)
					(q_1) edge                 node {$wx^Ny$} (q_2)
					(q_1) edge [loop above]    node {$v$} (q_1);
			\end{tikzpicture}
			\caption{
				Schematic of the automaton for Case 1 of the proof of Theorem \ref{context}.
				The number of states in the actual automaton is $\abs{uvwx^Ny}$.% + 1 - 1
			}\label{contextFig1}
		\end{figure}

		\begin{figure}
			\centering
			\begin{tikzpicture}[shorten >=1pt,node distance=2.5cm,on grid,auto]
				\node[state,initial]   (q_0)                {$q_0$}; 
				\node[state]           (q_1) [right of=q_0] {$q_1$}; 
				\node[state]           (q_2) [right of=q_1] {$q_2$}; 
				\node[state]           (q_3) [right of=q_2] {$q_3$}; 
				\node[state]           (q_4) [right of=q_3] {$q_4$}; 
				\path[->]
					(q_0) edge                 node {$u$}   (q_1)
					(q_1) edge                 node {$v^i$} (q_2)
					(q_2) edge                 node {$w$}   (q_3)
					(q_3) edge                 node {$y$}   (q_4)
					(q_1) edge [loop above]    node {$v^a$} (q_1)
					(q_3) edge [loop above]    node {$x^b$} (q_3);
			\end{tikzpicture}
			\caption{
				Schematic of the automaton for Case 2 of the proof of Theorem \ref{context}.
				The number of states in the actual automaton is $\abs{uv^av^iwx^by}-1$.% + 1 - 2
			}\label{contextFig}
		\end{figure}
			The case where $\abs{x}>\abs{v}$ is quite identical.
			The remainder of the proof concerns Case 2.

			\emph{Case 2:} $d:=\abs{v}=\abs{x}>0$.
			By Lemma \ref{bahnhof},
			for any positive integer $i$,
			the equation $ar+bs=i(a+b)$ has only the solution $r=s=i$ provided that
			$a$ and $b$ are relatively prime and both $a$ and $b$ are greater than $i$.
			In particular, this holds for any $a$ and $b$ with $a > i$ and $b = a + 1$.

			We construct an automaton $M$ as follows (Figure \ref{contextFig}).
			We put one loop of length $ad$ and later one of length $bd$,
			and add
			$id$ additional
			straggling states after the smaller loop of length $ad$.
			There are no loops apart from that. 

			Now for the analysis.
			Let $N = bi$.
			Each of the loops of $M$ will be traversed $i$ times during the processing of the word
			\[
				X_N = uv^{bi}wx^{bi}y.
			\]
			Let $U = \abs{u} + \abs{w} + \abs{y}$.
			Let us compare
			\[
				\abs{X_N} = n_N = U + 2bdi
			\]
			to the number of states of $M$,
			\[
				q = U + bd + ad + id - 1 = U + 2bd + (i-1)d - 1.
			\]
			(Note that when $i=1$, $q = (n_N + 1) - 2$ as expected as there are 2 repetitions of states.)
			In order to show $X_N\not\in L_3$ we need $q < \lfloor n_N/2\rfloor + 1$.
			To that end it suffices to have $q<n_N/2$, i.e.,
			\[
				bd + ad + id + U - 1 = (2b-1+i)d + U - 1 < \frac12(2bdi + U) = bdi + \frac{U}2.
			\]
			Equivalently,
			\[
				i - 1 + \frac{U}{2d} - \frac1d < (i - 2)b.
			\]
			Choose $i=3$; then the inequality will hold for all sufficiently large $b$.
			Thus, $M$ witnesses that $X_N\not\in L_3$, in contradiction to the pumping lemma (Theorem \ref{pump}).
		\end{proof}
	\section{Somewhat simple words do not form a \CFL}
		Let \RE\ denote the collection of
		recursively enumerable (or if you prefer, computably enumerable) languages.
		Recall that $L_3$ is the set of maximally complex words for nondeterministic automatic complexity over the
		alphabet $\{0,1,2\}$. Let $C$ denote plain Kolmogorov complexity and
		let
		\[
			R = \left\{x : C(x) \ge \abs{x}\right\}
		\]
		be the corresponding set of random words.
		We have seen (Theorem \ref{berstel} and Theorem \ref{context})
		that $L_3$ is \CFL-immune but not \co\CFL-immune.
		This is a pleasant analogue of
		the classical fact that $R$ is \RE-immune \cite[Section 3.1]{MR2732288}
		but not \co\RE-immune.
		Indeed, $R\in\co\RE$, but we conjecture that the analogous statement $L_3\in\co\CFL$ fails.
		\begin{con}\label{extra}
			$L_3\not\in\co\CFL$.
		\end{con}
		We shall confirm a variant of Conjecture \ref{extra} in Theorem \ref{canDo}.
		To that end, we need a couple of lemmas.
		\begin{lem}\label{subword}
			Let $A$ denote deterministic automatic complexity.
			Then $A(y)\le A(xyz)$ for all words $x,y,z$.
		\end{lem}
		\begin{proof}
			Given an automaton witnessing $A(xyz)$, we merely change the initial and final states
			to obtain an automaton witnessing an upper bound on $A(y)$.
		\end{proof}
		\begin{lem}\label{homomorphism}
			Let $A$ denote deterministic automatic complexity.
			Let $\pi:\{0,1\}^*\to\{0,1\}^*$ be an injective homomorphism with $\abs{\pi(0)} = \abs{\pi(1)}$.
			Then
			$A(x) \le A(\pi(x))$
			for each word $x$.
		\end{lem}
		\begin{proof}
			Let $M$ be a witnessing automaton for $A(\pi(x))$, with transition function $\delta$.
			We can now make an automaton $M'$ with the same states as $M$ (and the same initial and final states)
			that uniquely accepts $x$ among
			words of length $\abs{x}$ as follows.
			Throw out all the edges of $M$. Put an edge labeled $i$ from $q_1$ to $q_2$ in $M'$
			if $\delta(q_1,\pi(i))=q_2$ in $M$.

			It is clear that $M'$ accepts $x$. We now turn to uniqueness.

			Suppose $\abs{y}=\abs{x}$ and $M'$ accepts $y$.
			Since $\abs{\pi(0)}=\abs{\pi(1)}$, $\abs{\pi(y)}=\abs{\pi(x)}$.
			But $M$ only accepts one word of length $\abs{\pi(x)}$, so it must be that $\pi(y)=\pi(x)$.
			Since $\pi$ is injective, it follows that $y=x$.
		\end{proof}
		\begin{thm}[Shallit and Wang {\cite[Theorems 10 and 12]{MR1897300}}]\label{SW12}
			Let $A$ denote deterministic automatic complexity.
			There is a constant $n_0$ such that for $n\ge n_0$,
			\[
				\sqrt{n} - 1 \le A(0^n1^n)\le 6\sqrt n + 1.
			\]
		\end{thm}
		\begin{thm}\label{canDo}
			Let $\eps>0$, $f(x)=x^{1/2-\eps}$, and
			\[
				S = \{x \in \{0,1\}^* : A(x) < f(\abs{x})\}.
			\]
			Then $S\not\in\CFL$.
		\end{thm}
		\begin{proof}[Proof]
			We assume $S\in \CFL$ and derive a contradiction using the pumping lemma (Theorem \ref{pump}).
			Let $p$ be any sufficiently large pumping length.
			(The meaning of ``sufficiently large'' is determined below.)
			We shall build our unpumpable word as $s = r^k$ where
			\[
				r = 0^{p}1^{p}
			\]
			and $k$ is sufficiently large relative to $p$.
			Consider any decomposition of $s$ as $s = uvwxy$ where $\abs{vwx}\le p$ and $\abs{vx}\ge 1$.
			The main idea of the proof is the combination of the following two facts.
			\begin{itemize}
				\item $r$ is the shortest contiguous subword $R$ of $s$ such that the simplicity of $s$
				comes from repeating $R$.
				\item the ``pumpable part'' $vwx$ of $s$ is shorter than $r$.
			\end{itemize}
			 This means that pumping cannot help but
			increase the complexity. Thus by choosing $k$ wisely we will have $X_1\in S$ and
			$X_N := uv^Nwx^Ny\not\in S$ for some $N$, which will be a contradiction.

			The details are as follows.
			\begin{itemize}
				\item Case 1: $vx$ is all 0s.
					(We omit the case when $vx$ is all 1s as the proof is identical.)
					Then since $\abs{vwx}\le p$ and $\abs{vx}\ge 1$, the word $vwx$ is also all 0s.
				 	Let $N=pk$. We have $0^N={(0^p0^p)}^{k/2}$, and $v^Nwx^N$ contains $0^N$ as a contiguous subword.
					Now either
					at least half of the occurrences of $1^p$ as contiguous subwords of $X_N$ are in $u$, or
					at least half of them are in $y$.
					Hence $X_N$ contains a contiguous subword of
					the form either
					$(0^p0^p)^{k/2}(1^p0^p)^{k/2}$
					or
					$(0^p1^p)^{k/2}(0^p0^p)^{k/2}$.
					Then using the inequality $1\le\abs{vx}\le p$,
					\begin{equation}\label{3pk}
						3pk = n_1 + pk  = n_1 + N \le n_N
					\end{equation}
					and
					\begin{equation}\label{n_N}
						n_N = \abs{u v^N w x^N y}
						\le n_1 + Np = n_1 + p^2k = 2pk + p^2k.
					\end{equation}
					Now,
					\begin{eqnarray*}
						A(X_N)\ge &A(0^{k/2}1^{k/2})&\quad\text{by Lemma \ref{subword} and Lemma \ref{homomorphism}}\\
						 \ge &\sqrt{k/2} - 1&\quad\text{by Theorem \ref{SW12}}\\
						 \ge &\sqrt{n_N/(4p+2p^2)} -1 &\quad\text{by (\ref{n_N})}\\
						\ge &f(n_N)&
					\end{eqnarray*}
					provided
					\[
						(f(n_N)+1)^2/n_N\le \frac1{4p+2p^2}.
					\]
					Since $f$ is monotonically increasing, by (\ref{n_N}) it suffices that
					\[
						\frac{(f(2pk+p^2k)+1)^2}{n_N}\le \frac1{4p+2p^2}.
					\]
					For this, by (\ref{3pk}) it suffices that
					\[
						\frac{(f(2pk+p^2k)+1)^2}{3pk}\le \frac1{4p+2p^2}.
					\]
					In other words,
					\[
						\frac{(f((2p+p^2)k)+1)^2}{k}\le \frac{3p}{4p+2p^2} = \frac3{4+2p}.
					\]
					This is true for large enough $k$, since
					\[
						f(x) = o(\sqrt{x}).
					\]
					In order to guarantee that $A(X_1) < f(n_1)$,
					we require
					\[
						A(X_1) \le 2p+1 < f(n_1) = f(2pk).
					\]
					Since $f(x)\to\infty$ as $x\to\infty$, this holds by taking $k$ large enough.
				\item Case 2: $vx$ contains both 0s and 1s.
					Then $uv^Nwx^Ny$ contains
					a contiguous subword of the form $0^{N}1^{N}$
					(using the fact that the blocks $0^p1^p$ in $s$ are longer than the pumping length).
					The analysis is then similar to Case 1 (but without using Lemma \ref{homomorphism}).
			\end{itemize}
		\end{proof}
	\section{Automatic complexity of equivalence relations}
		We now go higher in the complexity-theoretic hierarchy, from $\CFL$ to $\NP$.
		We shall not be able to determine the $\NP$-completeness, or lack thereof, of problems like ``is $A(x)\le c$?''
		Nevertheless, we obtain results for a generalization of automatic complexity.
		\begin{df}\label{coheres with}
			Given a deterministic finite automaton (DFA)
			\[
				M=(Q,\Sigma,\delta,q_0,F),
			\]
			an equivalence relation $D$ on $Q$ \emph{induces}
			an equivalence relation $E$
			on a subset $S$ of $\{0,1\}^*$ if
			\[
				E = \{(x,y)\in S^2 \mid (\delta(q_0,x),\delta(q_0,y))\in D\}.
			\]

			A deterministic finite automaton
			$M=(Q,\Sigma,\delta,q_0,F)$
			\emph{coheres with} $E$ if
			there is an equivalence relation $D$ on $Q$ such that
			$D$ induces $E$.
		\end{df}
		In words, if $D$ induces $E$ then two words $x,y\in S$ are $E$-equivalent iff
		$M$ ends in $D$-equivalent states on input $x$ and on input $y$.

		Note that the set of final states $F$ is irrelevant in Definition \ref{coheres with}.
		\begin{df}\label{AE}
			The automatic complexity $A(E)$ of
			an equivalence relation $E$ is
			the least number of states of a DFA
			that coheres with $E$.
		\end{df}
		\begin{rem}
			Automatic complexity of a word $x$ (Shallit and Wang \cite{MR1897300})
			is a special case of automatic complexity of equivalence relations.
			Namely, the two equivalence classes are $\{x\}$ and
			$\{y: \abs{y}=\abs{x}, y\ne x\}$.
		\end{rem}

		\subsection{Complexity of equivalence relations is $\BH_2$-complete}
			As usual, let us say that a Boolean formula is \emph{CNF} if it is in conjunctive normal form, i.e., it is a conjunction of clauses,
			each of which is a disjunction of literals.
			\begin{df}[Encoding of literals]
				For a variable $x_j$, we denote the negation of $x_j$ by $\overline{x_j}$.
				We define
				\begin{eqnarray*}
					\neg^0x_j &=& x_j,\\
					\neg^1 x_j &=& \overline{x_j}.
				\end{eqnarray*}
				For a literal $l = \neg^b x_j$, where $b\in\{0,1\}$, we define the encoding word
				\[
					t(l) = 1^j0b0.
				\]
			\end{df}
			\begin{df}[inspired by {\cite[Victor Kuncak's solution to Exercise 7.36]{Sipser:13}}]\label{KuncakDef}
				Let $\phi$ be a CNF formula with $m$ clauses.
				Let
				\[
					Q_m = \{
						q_0,
						q_1,
						\dots,
						q_m,
						h,
						v_t,
						v_f,
						l_t,
						l_f,
						r,
						s
					\}
				\]
				be a set of cardinality $m+8$.
				For each $\sigma\in\{0,1\}^*$ and $q\in Q_m$,
				\[
					\sigma \to q
				\]
				is the ordered pair $\langle\sigma, q\rangle$.\footnote{See Remark \ref{https} for intuition.}
				Let $S$ be the following set, where $1^0=0^0=\lambda$, the empty word.
				\begin{eqnarray*}
					S &=& \{
					1^{m+1}\to q_0, 0\to h,\\
					&&00\to v_t, 01\to v_f,
					000\to l_t, 001\to l_f,
					010\to l_f, 011\to l_t,\\
					&&0^200\to s, 0^201\to r, 0^210\to q_0, 0^211\to r,\\
					&&0^300\to s, 0^301\to s, 0^310\to q_0, 0^311\to q_0\}\\
					&\cup&\{1^i\to q_i : 0\le i\le m\}\\
					&\cup&\{1^i011\to r : 1\le i\le m\}\\
					&\cup& \{t(l_1)t(l_2)t(l_3)\to s : (l_1\vee l_2\vee l_3)\text{ is a clause of }\phi\}.
				\end{eqnarray*}
				Let $E_\phi$ be the intersection of all equivalence relations containing
				\[
					\{
						(\sigma,\tau): (\exists q\in Q_m)
						((\sigma\to q)\in A\text{ and }(\tau\to q)\in S)
					\}.
				\]
			\end{df}
			\begin{rem}\label{https}
				The elements of $Q_m$ in Definition \ref{KuncakDef} are thought of as \emph{states}.
				The expression $\sigma\to q$ is to be thought of as the statement that $\delta(q_0,\sigma)=q$ where $\delta$ is the transition function of a DFA $M$ and $q_0$
				is the initial state.
				The equivalence relation $E_\phi$ identifies two words as equivalent if they lead us to the same state.
				Thus such an $M$ will cohere with $\phi$.
			\end{rem}
			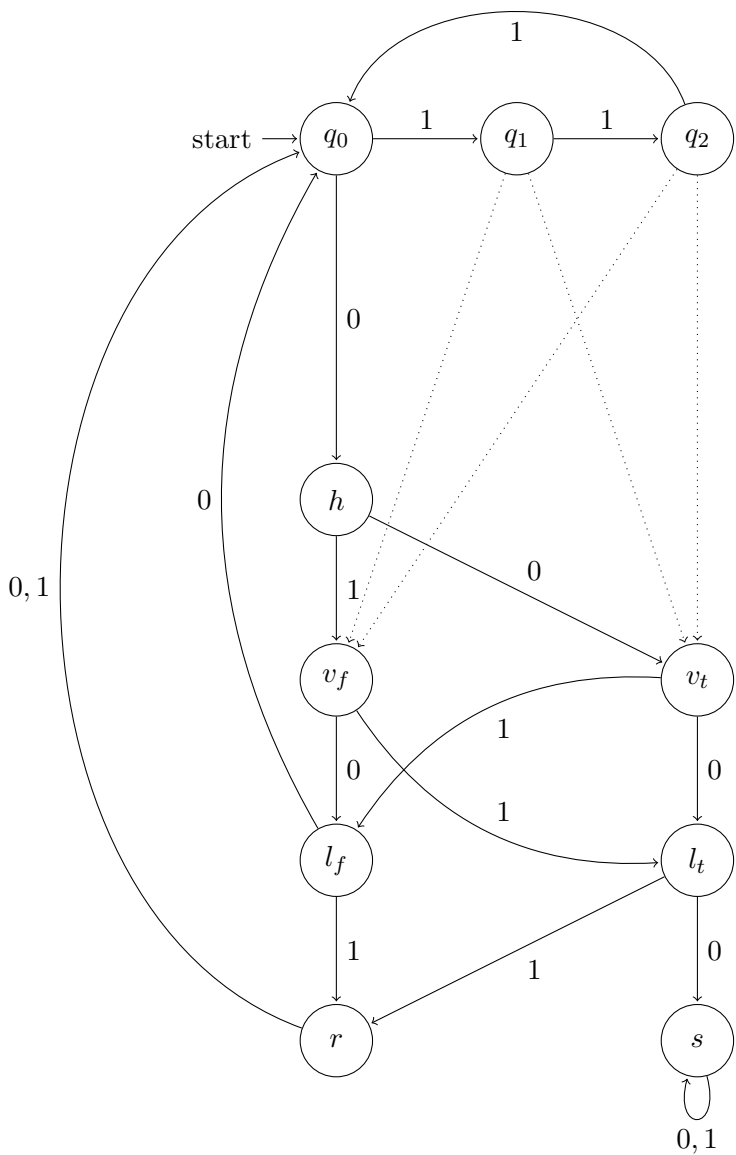
\begin{figure}
				\centering
				\begin{tikzpicture}[shorten >=1pt,node distance=2.4cm,on grid,auto]
					\node[state,initial]   (q_0)                {$q_0$}; 
					\node[state]           (q_1) [right of=q_0] {$q_1$}; 
					\node[state]           (q_2) [right of=q_1] {$q_2$}; 
					\node (j) [below of=q_0] {};
					\node[state]           (h)   [below of=j] {$h$};
					\node[state]           (v_f) [below of=h]   {$v_f$};
					\node                  (i)   [right of=v_f] {};
					\node[state]           (v_t) [right of=i]   {$v_t$};
					\node[state]           (l_f) [below of=v_f] {$l_f$};
					\node[state]           (l_t) [below of=v_t] {$l_t$};
					\node[state]           (r)   [below of=l_f] {$r$};
					\node[state]           (s)   [below of=l_t] {$s$};

					\path[->]
						(q_0) edge                 node {$1$}   (q_1)
						(q_1) edge                 node {$1$}   (q_2)
						(q_2) edge [bend right=70] node {$1$}   (q_0)
						(q_0) edge                 node {$0$}   (h)
						(h)   edge                 node {$1$}   (v_f)
						(h)   edge                 node {$0$}   (v_t)
						(v_f) edge                 node {$0$}   (l_f)
						(v_f) edge [bend right]    node {$1$}   (l_t)
						(v_t) edge [bend right]    node {$1$}   (l_f)
						(v_t) edge                 node {$0$}   (l_t)
						(l_f) edge [bend left]     node {$0$}   (q_0)
						(l_f) edge                 node {$1$}   (r)
						(r)   edge [bend left=70]  node {$0,1$} (q_0)
						(s)   edge [loop below]    node {$0,1$} (s)
						(l_t) edge                 node {$1$}   (r)
						(l_t) edge                 node {$0$}   (s)
						(q_1) edge [dotted]        node {}      (v_f)
						(q_1) edge [dotted]        node {}      (v_t)
						(q_2) edge [dotted]        node {}      (v_f)
						(q_2) edge [dotted]        node {}      (v_t);
				\end{tikzpicture}
				\caption{
					The automaton $M'$ from the proof of Theorem \ref{NPthm} is given by the solid lines.
					Appropriate choice of two of the dotted lines gives the total DFA $M$.
					The case where the formula $\phi$ has $m=2$ clauses is shown.
				}\label{benevolence}
			\end{figure}

			\begin{thm}\label{NPthm}
				$\{E : A(E) = \abs{E}\}$ is \NP-complete.
			\end{thm}
			\begin{proof}
				It is immediate from the definitions that
				\[
					\{E : A(E) = \abs{E}\} = \{E : A(E) \le \abs{E}\}.
				\]
				We reduce 3-\SAT\ to $\{E : A(E) \le \abs{E}\}$ using the mapping $\phi\mapsto E_\phi$ from
				Definition \ref{KuncakDef}.
				It induces a finite automaton $M'$ which is deterministic but whose transition function $\delta'$ is not total (Figure \ref{benevolence}).
				We see that $\phi$ is satisfiable iff there is a total DFA $M$, differing from $M'$ only in that its transition function $\delta\supseteq\delta'$
				is total, such that $M$ coheres with $E_\phi$. In particular $M$ has no more states than $M'$.
				The possible extra transitions of $M$ are shown in dotted lines in Figure \ref{benevolence}.
				Thus
				\begin{align*}
				 	\phi\text{ is satisfiable}  \quad &\Longrightarrow\quad A(E_\phi) \le \abs{E_\phi},\\
					\phi\text{ is unsatisfiable}\quad &\Longrightarrow\quad A(E_\phi) \not\le \abs{E_\phi}.
				\end{align*}
			\end{proof}
			\begin{thm}\label{malevolentThm}
				$\{E: A(E) = \abs{E} + 1\}$ is \coNP-hard.
			\end{thm}
			\begin{proof}
				It suffices to use the same reduction as in Theorem \ref{NPthm} and demonstrate unconditionally,
				i.e., without any assumption on satisfiability of $\phi$ or lack thereof, that
				\begin{align*}
					A(E_\phi) \le \abs{E_\phi} + 1.
				\end{align*}
				The question is then how to add one more state to Figure \ref{benevolence} to make the resulting automaton $M^+$ cohere with $E_\phi$.
				This is indicated in Figure \ref{malevolence}. We ensure $1^j011\rightarrow r$, i.e., $\delta(q_0, 1^j011) = r$, $1\le j\le m$
				using a new state $e$.
			\end{proof}
			In the proof of Theorem \ref{malevolentThm}, the state $e$ is acting duplicitously, in a sense, copying some of the behavior of the ``truth values'' $v_t$
			and $v_f$ without committing to a truth value.
			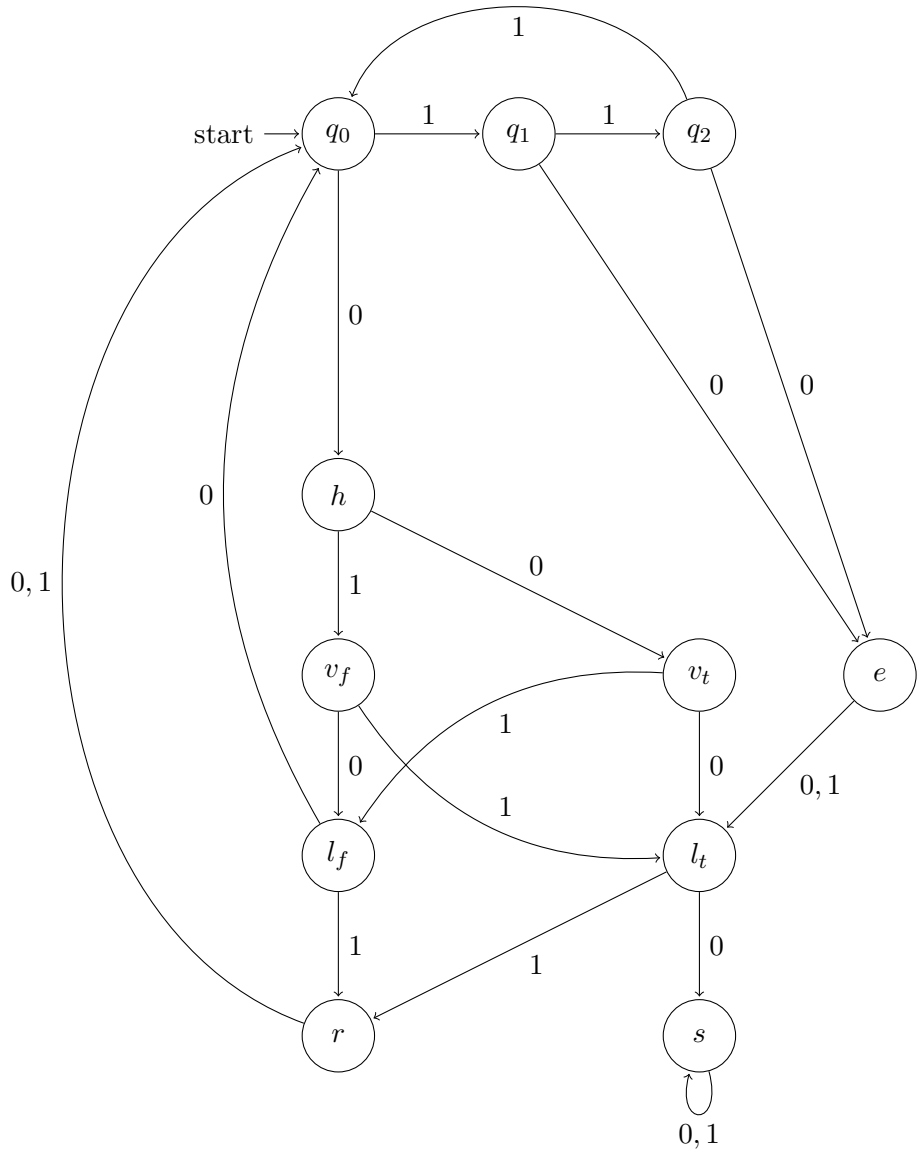
\begin{figure}
				\centering
				\begin{tikzpicture}[shorten >=1pt,node distance=2.4cm,on grid,auto]
					\node[state,initial]   (q_0)                {$q_0$}; 
					\node[state]           (q_1) [right of=q_0] {$q_1$}; 
					\node[state]           (q_2) [right of=q_1] {$q_2$}; 
					\node (j) [below of=q_0] {};
					\node[state]           (h)   [below of=j] {$h$};
					\node[state]           (v_f) [below of=h]   {$v_f$};
					\node                  (i)   [right of=v_f] {};
					\node[state]           (v_t) [right of=i]   {$v_t$};
					\node[state]           (e)   [right of=v_t] {$e$};
					\node[state]           (l_f) [below of=v_f] {$l_f$};
					\node[state]           (l_t) [below of=v_t] {$l_t$};
					\node[state]           (r)   [below of=l_f] {$r$};
					\node[state]           (s)   [below of=l_t] {$s$};

					\path[->]
						(q_0) edge                 node {$1$}   (q_1)
						(q_1) edge                 node {$1$}   (q_2)
						(q_2) edge [bend right=70] node {$1$}   (q_0)
						(q_0) edge                 node {$0$}   (h)
						(h)   edge                 node {$1$}   (v_f)
						(h)   edge                 node {$0$}   (v_t)
						(v_f) edge                 node {$0$}   (l_f)
						(v_f) edge [bend right]    node {$1$}   (l_t)
						(v_t) edge [bend right]    node {$1$}   (l_f)
						(v_t) edge                 node {$0$}   (l_t)
						(l_f) edge [bend left]     node {$0$}   (q_0)
						(l_f) edge                 node {$1$}   (r)
						(r)   edge [bend left=70]  node {$0,1$} (q_0)
						(s)   edge [loop below]    node {$0,1$} (s)
						(l_t) edge                 node {$1$}   (r)
						(l_t) edge                 node {$0$}   (s)
						(q_1) edge                 node {$0$}   (e)
						(q_2) edge                 node {$0$}   (e)
						(e)   edge                 node {$0,1$} (l_t);
				\end{tikzpicture}
				\caption{
					Automaton $M^+$ used in Theorem \ref{malevolentThm}. At the cost of adding a state $e$,
					we ensure that $M^+$ coheres with $E_\phi$, whether or not $\phi$ is satisfiable.
				}\label{malevolence}
			\end{figure}
			\begin{df}[Wechsung \cite{MR821265}]
				The first two levels of the Boolean hierarchy for \NP\ are given by
				\[
					\BH_1 = \NP,
				\]
				\[
					\BH_2 = \{ L_1 \setminus L_2 : L_1, L_2 \in \NP\}.
				\]
			\end{df}
			\begin{df}
				\[
					\SAT(2) = \{(\phi_1,\phi_2) : \phi_1\text{ is satisfiable and }\phi_2\text{ is not}\}.
				\]
			\end{df}
			\begin{thm}\label{cai}
				$\SAT(2)$ is complete for $\BH_2$ with respect to polynomial-time many-one reductions.
			\end{thm}
			Theorem \ref{cai} can be found in Cai et al.~\cite[Theorem 5.2]{MR972671}.
			(In their notation, $\BH_2 = \NP(2)$.)
			We will use without proof the extension of Theorem \ref{cai} from \SAT\ to 3-\SAT.
			\begin{thm}\label{exact}
				For each $k\ge 1$, $\{E : A(E) = \abs{E} + k\}$ is $\BH_{2}$-complete.
			\end{thm}
			\begin{proof}
				For each $k\ge 1$, $\{E: A(E) = \abs{E} + k\}$ equals
				\[
					\{E: A(E)\le \abs{E} + k\} \setminus \{E: A(E) \le \abs{E} +k - 1\} \in \BH_2.
				\]
				It remains to show $\BH_2$-hardness. Let $M_1$ and $M_2$ be automata as indicated in Figure \ref{malevolence}
				for two 3-\SAT\ instances $\phi_1$ and $\phi_2$, respectively.
				Let $k$ and $\ell\ne k$ be positive integers and let $\sigma_1,\dots,\sigma_{k+\ell}$ be incomparable words.
				We define $M$ in a natural way so that
				\[
					L(M) = \bigcup_{i=1}^{\ell}\{\sigma_i\, x: x\in L(M_1)\} \cup \bigcup_{i=\ell+1}^k\{\sigma_i\,x: x\in L(M_2)\}.
				\]
				Note that the various ``$e$'' states for distinct copies of $M_1$ and $M_2$ must be distinct, since they transition to distinct ``$l_t$'' states.
				\begin{table}
					\centering
					\begin{tabular}{|c|c|c|}
						\hline
						$\phi_1$ satisfiable? & $\phi_2$ satisfiable? & Number of extra states \\
						\hline
						no & no & $\ell + k$\\
						no & yes& $\ell$\\
						yes& no & $k$\\
						yes&yes & $0$\\
						\hline
					\end{tabular}
					\caption{The number of extra states needed for Theorem \ref{exact}.}\label{summer-of}
				\end{table}
				Thus, we have
				\begin{quote}
					$\phi_1$ is satisfiable, but $\phi_2$ is not
				\end{quote}
				if and only if
				\begin{quote}
					we need no extra state for $\phi_1$, but $k$ extra states for $\phi_2$,
				\end{quote}
				if and only if (by Table \ref{summer-of})
				\begin{quote}
					we need $k$ extra states overall,
				\end{quote}
				if and only if
				$A(E) = \abs{E} + k$.
				We make corresponding changes in the ``axioms'' (the elements of $S$ in Definition \ref{KuncakDef}) for $E_{\phi_i}$.
				Applying Theorem \ref{cai} completes the proof.
			\end{proof}
	\bibliographystyle{plain}
	\bibliography{complexity-of-automatic-complexity}
\end{document}